\documentclass[a4paper,11pt]{article}       

\usepackage{graphicx}
\usepackage{latexsym}
\usepackage{amssymb,amsthm}
\usepackage{amsmath}
\usepackage{hyperref}

\newcommand{\Rset}{\ensuremath{\mathbb{R}}}

\newcommand{\Rsetnn}{\ensuremath{\mathbb{R}_{\geq 0}}}
\newcommand{\sign}{\operatorname{sign}}
\newcommand{\rank}{\operatorname{rank}}
\newcommand{\Img}{\operatorname{Img}}
\newcommand{\sigGI}{\omega}
\renewcommand{\vec}{\mathbf}

\theoremstyle{plain}
\newtheorem{theorem}{Theorem}[section]
\newtheorem{lemma}[theorem]{Lemma}
\newtheorem{proposition}[theorem]{Proposition}
\newtheorem{corollary}[theorem]{Corollary}
\theoremstyle{definition}
\newtheorem{definition}[theorem]{Definition}
\newtheorem{example}[theorem]{Example}

\begin{document}
\title{A unified view on bipartite species-reaction and interaction graphs for chemical reaction networks}

\author{Hans-Michael Kaltenbach\\
              Dep. Biosystems Science and Engineering, ETH Zurich\\
              \url{hans-michael.kaltenbach@bsse.ethz.ch}}

\maketitle

\begin{abstract}
The Jacobian matrix of a dynamic system and its principal minors play a prominent role in the study of qualitative dynamics and bifurcation analysis. When interpreting the Jacobian as an adjacency matrix of an interaction graph, its principal minors correspond to sets of disjoint cycles in this graph and conditions for various dynamic behaviors can be inferred from its cycle structure. For chemical reaction systems, more fine-grained analyses are possible by studying a bipartite species-reaction graph. Several results on injectivity, multistationarity, and bifurcations of a chemical reaction system have been derived by using various definitions of such bipartite graph. Here, we present a new definition of the species-reaction graph that more directly connects the cycle structure with determinant expansion terms, principal minors, and the coefficients of the characteristic polynomial and encompasses previous graph constructions as special cases. This graph has a direct relation to the interaction graph, and properties of cycles and sub-graphs can be translated in both directions. A simple equivalence relation enables to decompose determinant expansions more directly and allows simpler and more direct proofs of previous results.
\end{abstract}

\section{Introduction}
A major problem in the analysis of chemical reaction systems is the fact that parameters such as kinetic rate constants for are inherently difficult to obtain from experimental data, and \textit{in-vitro} data might not be valid for \textit{in-vivo} systems. Moreover, the exact algebraic form of chemical rate laws is often difficult to determine. In contrast to more general dynamic systems, however, the dynamics of a chemical reaction system is additionally restricted by the stoichiometry and topology of the underlying reaction network, and chemical rate laws are typically monotone in the concentrations. The underlying reaction network also induces algebraic dependencies of state variables and parameters. Several approaches have been proposed that exploit these additional constraints and allow to determine ---without knowledge of parameter values and minimal conditions on the rate laws--- if a system is capable of certain qualitative dynamics such as oscillations and multiple equilibria, and to establish stability of equilibria. These methods fall broadly into three categories: 

(i) methods based on the particular algebraic structure of reaction systems include the deficiency-based Chemical Reaction Network Theory~\cite{horn1972,feinberg1995,feinberg1995a}, the Stoichiometric Network Analysis~\cite{clarke1980}, and Biochemical Systems Theory~\cite{savageau1969,savageau1969a,savageau1970}.

(ii) methods based on the cycles of a signed interaction graph derived from the Jacobian: e.g., a positive cycle is necessary for multistationarity, whereas a negative cycle is necessary for oscillations~\cite{gouze1998,thomas2001,soule2003,kaufman2007}. Other criteria have also been investigated~\cite{domijan2011}. The absence of negative undirected cycles implies monotonicity with respect to an orthant cone, ruling out chaotic and oscillatory dynamcis~\cite{sontag2007,smith1988,hirsch2005}; 

(iii) methods based on a bipartite graph with vertices for both species and reactions been proposed recently to study criteria for the existence multiple equilibria or oscillations~\cite{banaji2008,mincheva2007,craciun2005,craciun2006a} and for establishing monotonicity~\cite{sontag2007,angeli2010}. Exploiting the relation of principal minors of a matrix and cycles in its associated graph to coefficients in the characteristic polynomial, several criteria were also formulated to infer the possibility of saddle-point and Hopf bifurcations in a system~\cite{mincheva2007}. Definitions of such graphs in the context of qualitative dynamics include undirected~\cite{craciun2006a,banaji2008} and directed~\cite{mincheva2007} species-reaction (SR) graphs as well as graphs with special edge-types~\cite{banaji2008}. For various proofs, oriented versions of an undirected SR-graph also need to be considered in~\cite{craciun2006a}. 

In this work, our goal is to elucidate how the various interaction- and bipartite-graph-based methods are related and to propose a new version of the species-reaction graph, named directed species-reaction graph (DSR-graph), that immediately shows its relation to the Jacobian matrix and gives more direct justification and proofs of several theorems. Using this graph, we present refined criteria for sign-definiteness of the determinant of the Jacobian and its principal minors and answer a question raised in~\cite{mincheva2007} of the equivalence of criteria developed independently in~\cite{mincheva2007} and~\cite{craciun2006a}. We propose a simple equivalence relation on the DSR-graph that allows to study determinants by studying each of its equivalence classes. We also develop and emphasize the relation between determinant expansions, the Jacobian matrix, and the DSR-graph, making use of long-known relations~\cite{harary1962}. 

\section{Chemical Reaction Systems}
We consider a biochemical reaction network with $n$ \emph{species} $S_1,\dots, S_n$ and $r$ \emph{reactions} $R_1,\dots,R_r$. A reaction $R_j$ is formally given by
\[
  R_j: \sum_{i=1}^n y_{i,j} S_i \to \sum_{i=1}^n y'_{i,j} S_i\;,
\]  
where $y_{i,j},y'_{i,j}\in\Rsetnn$ are the \emph{molecularities} of the substrates and products, respectively. The matrix $\vec{N}\in\Rset^{n\times r}$ with entries $N_{i,j}=y'_{i,j}-y_{i,j}$ is called the \emph{stoichiometric matrix}, its $j$th column vector $(y'_{1,j}-y_{1,j},\dots,y'_{n,j}-y_{n,j})^T\in\Rset^n$ is the \emph{stoichiometry} of reaction $R_j$. We require that all reactions are irreversible; this can always be achieved by  splitting reversible reactions into a forward- and a backward reaction. We also require that each species occurs on at most one side in each reaction, i.e., we exclude reactions such as $A+B\to 2A$. The signs of all $N_{i,j}$ are then uniquely defined such that $N_{i,j}>0$ (resp. $N_{i,j}<0$) if species $S_i$ is produced (resp.\ consumed) by reaction $R_j$, and $N_{i,j}=0$ if it does not participate in the reaction.

We denote by $x_i(t)\in\Rsetnn$ the concentration of species $S_i$ at time $t$, and collect these concentrations in the \emph{state vector} $\vec{x}\equiv\vec{x}(t)=(x_1(t),\dots,x_n(t))^T$, adopting the usual convention to drop the explicit dependence on $t$. Each reaction $R_j$ is assigned a \emph{rate law}
\[
  v_j :  \Rsetnn^n  \to  \Rsetnn\;,
\]
describing the velocity of the reaction as a function of the current state $\vec{x}$ of the system. A typical rate law is the mass-action law $v_j(\vec{x}) = k_j\cdot x_1^{y_{1,j}}\cdots x_n^{y_{n,j}}$ with some reaction-specific rate constant $k_j>0$. The rate laws are collected into the \emph{rate vector} $\vec{v}(\vec{x})=(v_1(\vec{x}),\dots,v_r(\vec{x}))^T$, leading to the system of non-linear differential equations
\begin{equation}
  \label{eq:xNv}
  \frac{d}{d t}\vec{x} = \vec{f}(\vec{x}) = \vec{N}\cdot\vec{v}(\vec{x})\;,
\end{equation}
which govern the evolution of species-concentrations over time.

A special class of rate laws is given by the \emph{non-autocatalytic (NAC)} laws, for which $\partial v_j/\partial x_i > 0$ when specie $S_{i}$ is a substrate of reaction $R_{j}$ and all concentrations are positive, and $\partial v_j/\partial x_i \equiv 0$ if $S_{i}$ is not a substrate of $R_{j}$. These conditions are quite natural: they require that an increase in a reactant's concentration cannot decrease the reaction rate and that species that are not reactants of a reaction do not directly influence its rate. The class NAC contains both mass-action and Michaelis-Menten rate laws, as well as many others. Very similar conditions for reversible reactions were given in~\cite{banaji2008}.

To model mass transport over the boundaries of the system, we also allow reactions of the form $R_j: \emptyset\to S_i$ with constant rate $v_j(\vec{x})=k_j>0$ to represent \emph{production} or \emph{inflow} of $S_i$, and of the form $R_j: S_i\to \emptyset$ with rate $v_j(\vec{x})=k_j\cdot x_i$ to represent \emph{degradation} or \emph{outflow} of $S_i$.

A concentration $\vec{x}^*\in\Rsetnn^n$ is called an \emph{equilibrium} or a \emph{steady state} if
\[
  \vec{f}(\vec{x}^*) = \vec{0}\;;
\]
it is called \emph{positive}, if $x_i^*>0$ for all $i=1,\dots,n$, which we abbreviate $\vec{x}>\vec{0}$. With initial condition $\vec{x}_0=\vec{x}(0)$, all solutions of system~(\ref{eq:xNv}) are confined to the affine subspace $\vec{x}_0+\Img\vec{N}$ of dimension $\rank(\vec{N})$. A system is said to have \emph{multiple positive equilibria} or to be \emph{multistationary}, if there are $\vec{x}^*\not=\vec{x}^{**}>\vec{0}$ such that
\begin{align*}
  \vec{f}(\vec{x}^*) &= \vec{0} \\
  \vec{f}(\vec{x}^{**}) &= \vec{0} \\
  \vec{x}^{**}-\vec{x}^* &\in \Img\vec{N}\;.
\end{align*}
Without loss of generality, we may assume that the system has been reduced such that $\vec{N}$ has full rank; then, $\Img\vec{N}$ coincides with the state space and the third condition for multistationarity is trivially fulfilled.

\begin{example}
\label{ex:small}
We consider the following reaction network
\begin{eqnarray*}
  R_1: &aA \to& bB \\
  R_2: &cB \to& dA \\
  R_3: &A \to& \emptyset\\
  R_4: &\emptyset \to& A\\
  R_5: &B \to& \emptyset\\
  R_6: &\emptyset \to& B\;.
\end{eqnarray*}  
This network has $n=2$ species $A,B$ with concentrations $\vec{x}=(x_A,x_B)^T$ and $r=6$ reactions $R_1,\dots,R_6$, where reactions $R_3$ and $R_5$ describe the degradation or outflow of species $A,B$, respectively, and reactions $R_4$ and $R_6$ their production or inflow into the system. Reactions $R_1$ and $R_2$ are the \emph{internal} or \emph{true} reactions. The stoichiometry of reaction $R_1$ is determined by $y_{1,1}=a, y_{2,1}=0$, $y'_{1,1}=0, y'_{2,1}=b$ and is thus $(-a, b)^T$. The stoichiometric matrix is
\[
  \vec{N} = \begin{pmatrix}
     -a & d & -1 & 1 & 0 & 0 \\
     b & -c & 0 & 0 & -1 & 1  
  \end{pmatrix}
\]
with $\rank(\vec{N})=2$. The dynamics is given by the two differential equations
\begin{eqnarray*}
  (d/dt) x_A &=& -a\, v_1(\vec{x}) + d\, v_2(\vec{x}) - v_3(\vec{x}) + v_4(\vec{x}) \\
  (d/dt) x_B &=& -c\, v_2(\vec{x}) + b\, v_1(\vec{x}) - v_5(\vec{x}) +  v_6(\vec{x})\;.
\end{eqnarray*}
The production rates $v_4(\vec{x})\equiv k_4$ and $v_6(\vec{x})\equiv k_6$ are constant.
\end{example}

\section{Qualitative Dynamics and the Interaction Graph and }
Several criteria for qualitative dynamics such as oscillations or multistationarity can be established via properties of the Jacobian matrix of a reaction system. The Jacobian can be interpreted in terms of a graph and its principal minors then correspond to certain cycle structures in this graph. Of particular interest are conditions such that these cycles are sign-definite, that is, do not depend on the concentrations $\vec{x}$ at which the Jacobian is evaluated.

The \emph{Jacobian matrix} of a reaction system (\ref{eq:xNv}) is the function
\[
  \vec{J} = \left(\frac{\partial f_i}{\partial x_j}\right)_{1\leq i,j\leq n} = \vec{N}\cdot\left(\frac{\partial v_i}{\partial x_j}\right)_{1\leq i\leq r, 1\leq j\leq n}\;.
\]
Evaluated at $\vec{x}_0$, the resulting matrix $\vec{J}(\vec{x}_0)$ allows to approximate the system's dynamics in the vicinity of $\vec{x}_0$.

A tupel $\alpha=(\alpha_1,\dots,\alpha_l)\subseteq\{1,\dots,n\}$ with $\alpha_1<\cdots<\alpha_l$ is called a \emph{multi-index} of $\{1,\dots,n\}$ of size $|\alpha|:=l$. Let $\vec{A}\in\Rset^{n\times m}$ and let $\alpha$ and $\beta$ be two multi-indices, not necessarily of the same size, of $\{1,\dots,n\}$ and $\{1,\dots,m\}$, respectively. The matrix
\[
  \vec{A}_{\alpha,\beta} := (a_{i,j})_{i\in\alpha,\, j\in\beta}
\]
is the $|\alpha|\times|\beta|$ sub-matrix of $\vec{A}$ derived by removing all rows and columns with indices not in $\alpha$ and $\beta$, respectively. We denote $\vec{A}_\alpha:=\vec{A}_{\alpha,\alpha}$ for brevity. The determinant $\det(\vec{A}_\alpha)$ is called a \emph{principal minor} of order $|\alpha|$.

The \emph{characteristic polynomial} of $\vec{J}$ is given by
\[
  P_{\vec{J}}(\lambda) = \det(\lambda \vec{I} - \vec{J})
  = \lambda^n+c_{n-1}\lambda^{n-1}+\cdots+c_1\lambda+c_0\;.
\]  
Its coefficients $c_i$ can be computed by summing over all principal minors of $\vec{J}$ of order $n-i$:
\begin{equation}
\label{eq:charpolcoeff}
  c_i = (-1)^{n-i}\sum_{\alpha\subseteq\{1,\dots,n\} \atop |\alpha|=n-i} \det(\vec{J}_\alpha)
      = \sum_{\alpha\subseteq\{1,\dots,n\} \atop |\alpha|=n-i} \det(-\vec{J}_\alpha)\;,
\end{equation}
where we set $c_n=1$ for completeness. Two well-known special cases are $c_0=\det(-\vec{J})$ and $c_{n-1}=\text{tr}(-\vec{J})$.

The roots of $P_{\vec{J}}$ are the eigenvalues of the system. A \emph{saddle-node bifurcation}, at which an equilibrium point splits into two, requires a single zero eigenvalue. The condition $c_0=0$ and therefore the vanishing of the determinant in at least one point is a necessary condition for this bifurcation. A \emph{Hopf-Andronov bifurcation}, at which an equilibrium changes into a limit cycle, requires a single pair of conjugate eigenvalues with zero real part. Under certain conditions, the vanishing of coefficients $c_i=0$ for $i\not= 0$ implies that the Hurwitz determinant of order $(n-1)$ vanishes, which in turn is a necessary condition for a conjugate pair of imaginary eigenvalues~\cite{mincheva2007}. These observations give a direct connection of conditions for bifurcations and (non-)vainishing of principal minors via equations~(\ref{eq:charpolcoeff}). For fully open systems with inflow rates $\vec{K}$ and internal and outflow reactions given by $\vec{N}$, the condition that two positive equilibria match the same inflow rates is given by
\[
  \vec{K} = -\vec{N}\vec{v}(\vec{x^*}) = -\vec{N}\vec{v}(\vec{x^{**}})\;;
\]
multistationarity can be excluded if $-\vec{N}\vec{v}(\vec{x})$ is \emph{injective}. For mass-action kinetics, this requires that $\det(-\vec{J})>0$ for all $\vec{x}>\vec{0}$~\cite{craciun2006}, for NAC kinetics, the same has to hold for all its principal minors~\cite{banaji2010}

Before introducing the interaction graph itself, we briefly recall some standard definitions for arbitrary directed graphs $G=(V,E,\gamma)$ with vertex-set $V$, edge-set $E$, and edge-label function $\gamma:E\to\Rset$: an edge $e=(u,\,u')\in E$ starting in $u$ and ending in $u'$ is said to be \emph{incident} to either vertex and conversely either vertex is incident to $e$.

A \emph{path of length $q$} or \emph{$q$-path} is a sequence of edges $((u_1,\,u_2),\,(u_2,\,u_3),\dots,(u_{q-1},\,u_q))$ such that $(u_i,\,u_{i+1})\in E$ for $1\leq i<q$. We will alternatively also denote it by its sequence of vertices $(u_1,\dots,u_q)$. It is a \emph{simple path} if $u_i\not=u_j$ for $i\not=j$. It is a \emph{cycle} $C$ of length $q$, if $u_1=u_q$ and a \emph{simple cycle} or \emph{circuit} if it is a cycle and a simple path. A cycle $(u,\,u)\in E$ is a \emph{self-loop} (of length $q=1$).  

A \emph{sub-graph} $H=(V(H),\, E(H))$ of $G$, denoted $H\subseteq G$, is a graph with vertices $V(H)\subseteq V$, edges $E(H)\subseteq (V(H)\times V(H))\cap E$ and edge-label function $\gamma|_{E(H)}$. We define the label of the whole sub-graph $H$ to be
\[
  \gamma(H):=\prod_{e\in E(H)} \gamma(e) \;.
\]
A sub-graph $H$ is \emph{positive} (\emph{negative}) if $\gamma(H)>0$ ($\gamma(H)<0$). Two sub-graphs $H$ and $H'$ are \emph{disjoint} if their vertex and edge sets are. The union, intersection, and difference of (sub-)graphs $H,H'$ are defined by the corresponding operations on vertex and edge-sets and restriction of $\gamma$ to the resulting sets, e.g., $H\backslash H'$ is the graph given by
\begin{eqnarray*}
   V(H\backslash H') &=& V(H)\backslash V(H') \\
   E(H\backslash H') &=& \left(E(H)\backslash E(H')\right)
     \;\cap\;\left(V(H\backslash H')\times V(H\backslash H')\right)
\end{eqnarray*}
and edge-label function $\gamma|_{E(H\backslash H')}$. Moreover, $H+H'$ denotes the direct sum, i.e., the union of disjoint sub-graphs.

The interaction graph is constructed by interpreting the Jacobian matrix as the (labelled) adjacency matrix of a directed graph.
\begin{definition}[interaction graph $G_I$]
Let $\vec{J}$ be the Jacobian matrix of a chemical reaction system~(\ref{eq:xNv}) with species $S_1,\dots,S_n$. The \emph{interaction graph} $G_I=G_I(\vec{J})=(V,E,\gamma)$ is the directed edge-labelled graph given by
\begin{align*}
  V &:= \{S_1,\dots,S_n\} \\
  E &:= \{(S_i,\,S_j)\in V\times V \,|\, J_{ji}\not\equiv 0\} \\
  \gamma(e) &:= J_{ji} \text{ for } e=(S_i,\,S_j)\in E\;.
\end{align*}
\end{definition}
Note that the topology of $G_I$ is independent of the species-concentrations $\vec{x}$, but that the edge-labels are functions in $\vec{x}$.

\begin{definition}[sign-definite sub-graph]
A sub-graph $H\subseteq G_I$ is \emph{sign-definite} if
\[
  \sign(\gamma(H)(\vec{x}_0)))
\]
is independent of $\vec{x}_0$.
\end{definition} 
A sub-graph is sign-definite if each of its edges is; the converse is not necessarily true. Under our assumptions on the network, sign-definiteness of an edge in $G_I$ can be established from the stoichiometry alone.
\begin{lemma}[sign-definite edges]
\label{lem:sign-definite}
Let $G_I$ be an interaction graph. An edge $(S,S')\in E(G_I)$ is sign-definite if and only if there are no two reactions $R,R'$ such that $S$ and $S'$ are reactants in $R$ and reactant and product in $R'$, respectively.
\end{lemma}
\begin{proof}
W.l.o.g., let $S=S_i$ and $S'=S_j$. We have
\[
  \gamma((S_i,S_j)) \;=\; J_{j,i} = \sum_{l=1}^r N_{j,l}\,\frac{\partial v_l}{\partial x_i}\;.
\]
We look at each individual reaction $R_l$: if $S_i$ is not a reactant of $R_l$, then $\partial v_l/\partial x_i\equiv 0$ due to NAC kinetics. Assume $S_i$ is a reactant, then $\partial v_l/\partial x_i > 0$ and $N_{i,l}<0$. Then, either $S_j$ is also a reactant, in which case $N_{j,l}<0$ and the contribution of $R_l$ is negative, or $S_j$ is a product, in which case $N_{j,l}>0$ and the contribution is positive. Thus, the overall sign of $\gamma((S_i,S_j))$ is indefinite if and only if $S_j$ is a reactant in one and a product in another reaction, which can be established from the stoichiometric matrix alone.
\end{proof}

The proof of Lemma~\ref{lem:sign-definite} directly suggests a method to make each edge sign-definite by splitting it into a positive and a negative edge. Important features of qualitative dynamics are preserved by appropriate choice of constants~\cite{helton2010,sontag2007}. However, if a sign-indefinite edge is part of a cycle, this procedure will result in a graph with two cycles of opposite signs.

\begin{example}
\label{ex:smallJac}
The system of Example~\ref{ex:small} has Jacobian
\[
  \vec{J} = \begin{pmatrix}
    -a\frac{\partial v_1}{\partial A} - \frac{\partial v_3}{\partial A} & d\frac{\partial v_2}{\partial B} \\
    b\frac{\partial v_1}{\partial A} & -c\frac{\partial v_2}{\partial B} - \frac{\partial v_5}{\partial B}
  \end{pmatrix}\;,
\]
where we denoted $x_A\equiv A$ and $v_i\equiv v_i(x_A,x_B)$ etc.\ for brevity.

\begin{figure}[htbp]
\begin{center}
\includegraphics[width=0.7\textwidth]{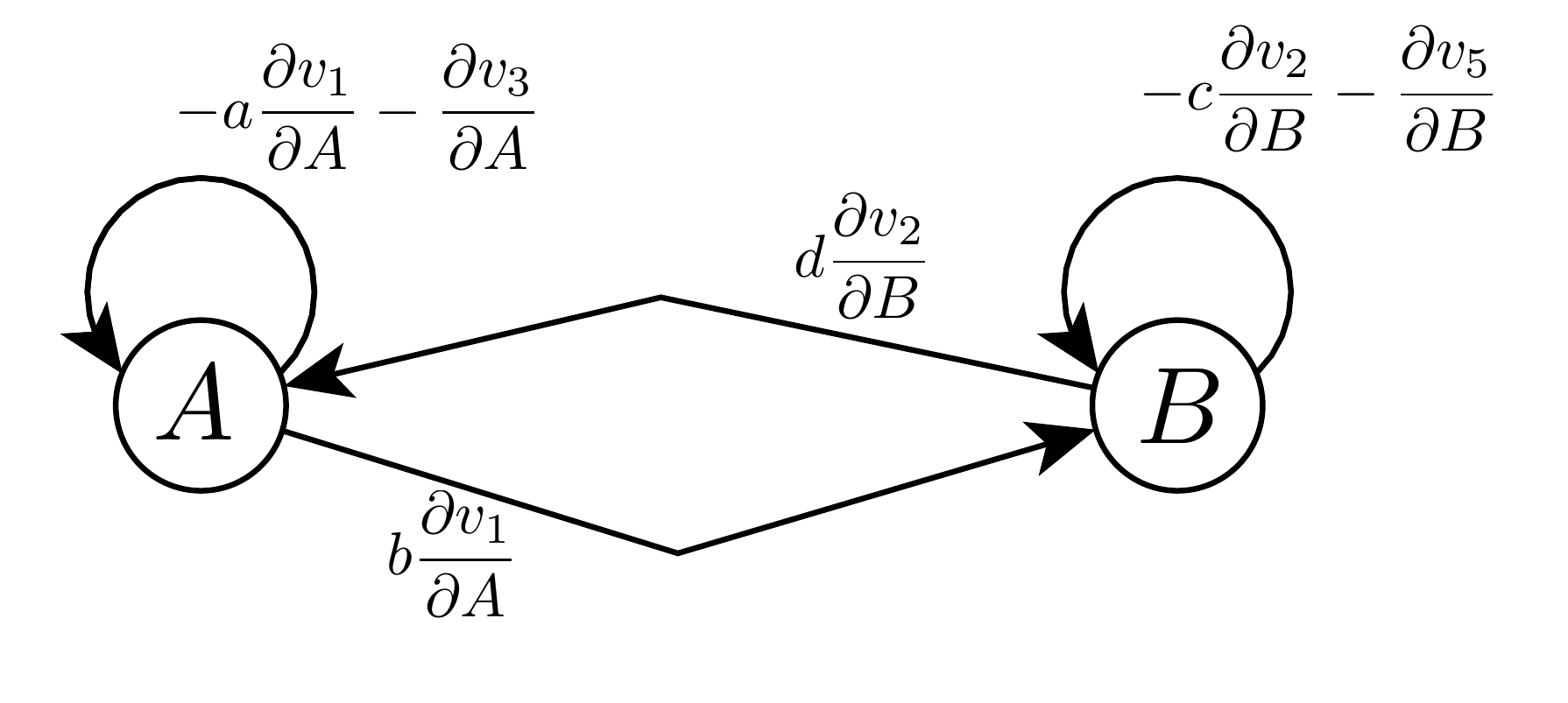}
\caption{Labelled interaction graph $G_I$ of example network. Note that production reactions $R_4$ and $R_6$ have zero partial derivatives and are thus not part of the graph.}
\label{fig:interaction}
\end{center}
\end{figure}

The corresponding interaction graph $G_I$ is given in Figure~\ref{fig:interaction}. It contains three simple cycles: the two self-loops $C_1=(A,A)$ with $\gamma(C_1)=-a\frac{\partial v_1}{\partial A} - \frac{\partial v_3}{\partial A}$, respectively $C_2=(B,B)$ with $\gamma(C_2)=-c\frac{\partial v_2}{\partial B} - \frac{\partial v_5}{\partial B}$, each of length 1, and the cycle $C_3 = (A, B, A)$ with $\gamma(C_3)=d\frac{\partial v_2}{\partial B}\cdot b\frac{\partial v_1}{\partial A}$ of length 2. The sign of each edge is independent of $\vec{x}_0$ and $G_I$ is therefore sign-definite.  
\end{example}

\section{Determinant Expansions and Characteristic Polynomial}
Recall that for a matrix $\vec{A}=(a_{i,j})\in\Rset^{n\times n}$, the determinant expansion is given by
\begin{equation}
\label{eq:detexpansionSn}
  \det(\vec{A}) = \sum_{\pi} (-1)^{\sign(\pi)}\prod_{i=1}^n a_{i,\pi(i)}\;,
\end{equation}
where $\pi$ runs over the permutation group on $\{1,\dots,n\}$ and $\sign(\pi)$ is its sign. An equivalent expansion can be formulated in purely graph-theoretic terms using line-graphs.
\begin{definition}[line-graph~\cite{harary1962}]
Let $C_1,\dots,C_q$ be a collection of disjoint simple cycles $C_i$ covering each vertex of $G$ exactly once:
\[
  V(G) = V(C_1) \cup \cdots \cup V(C_q)\;\text{and}\;V(C_i)\cap V(C_j)=\emptyset,\,i\not=j\;.
\]
Their union
\[
  L := \bigcup_{i=1}^q C_i \subseteq G
\]
is called a \emph{line-graph} of $G$. We denote by $\mathcal{L}(G)$ the collection of all line-graphs of $G$.
\end{definition}
A line-graph is also called \emph{Hamiltonian hooping}~\cite{soule2003} or \emph{nucleus}~\cite{domijan2011}, and the special definition of a \emph{subgraph} in~\cite{mincheva2007} relates to the same concept.

We again interpret $\vec{A}$ as the adjacency matrix of a labelled directed graph $G=G(\vec{A})$ with vertices $V=\{u_1,\dots,u_n\}$ and edge label function $\gamma((u_i,u_j))=A_{j,i}$. For any permutation $\pi$, the term $\prod_{i=1}^n A_{i,\pi(i)}$ is nonzero if and only if all corresponding edges $(u_{\pi(i)},u_i)$ exist. These edges induce a line-graph $L$ of $G$ and the product is the label of this line-graph: $\gamma(L)=\prod A_{i,\pi(i)}$. The sign of the permutation is also readily extracted from the graph.

\begin{definition}[signum of a sub-graph]
Let $H$ be any sub-graph of a directed graph $G$. Let $\epsilon(H)$ be the number of even-length cycles in $H$. We call the number
\[
  \sigGI(H) := (-1)^{\epsilon(H)}\;.
\]
the \emph{signum} of $H$.
\end{definition}
The signum $\sigGI(L)$ of a line-graph $L$ is the sign of the permutation described by $L$~\cite{harary1962}, is independent of $\sign(\gamma(H))$, and is directly related to the signum $\xi(L)$ proposed in~\cite{soule2003} via $\sigGI(L)=-\xi(-L)$, where $-L$ denotes the line-graph with edge-labels $-\gamma(\cdot)$.

\begin{lemma}[Harary~\cite{harary1962}]
\label{lem:detexpansionGI}
Let $\vec{A}\in\Rset^{n\times n}$ with graph $G=G(\vec{A})$. A determinant expansion of $\vec{A}$ is then given by
\begin{equation}
\label{eq:detexpansionGI}
  \det(\vec{A}) = \sum_{L\in\mathcal{L}(G)} \sigGI(L)\,\gamma(L)
    = \sum_{L\in\mathcal{L}(G)} \prod_{C\subseteq L} \sigGI(C)\, \gamma(C)\;.
\end{equation}   
\end{lemma}

We will focus our discussion on $G_I(\vec{J})$, but remark that all arguments also hold for arbitrary square sub-matrices $\vec{J}_\alpha$ by using the corresponding sub-graph $G_I(\vec{J}_\alpha)\subseteq G_I(\vec{J})$ induced by $V(G(\vec{J}_\alpha))=\{V_i\,|\, V_i\in V(G_I),\,i\in \alpha\}$. Each principal minor corresponds to a set of line-graphs of $G(\vec{J}_\alpha)\subseteq G(\vec{J})$, but each line-graph is associated to exactly one principal minor and we can compute each coefficient $c_i$ of the characteristic polynomial by investigating the line-graphs of all sub-graphs $G(\vec{J}_\alpha)$ with $|\alpha|=n-i$.

\begin{example}
\label{ex:smallcont}
Consider again the system of Example~\ref{ex:small}. Using expansion via permutations, the Jacobian matrix has determinant
\[
  \det(\vec{J}) \;=\;
  -bd\frac{\partial v_1}{\partial A}\frac{\partial v_2}{\partial B}  + 
    ac\frac{\partial v_1}{\partial A}\frac{\partial v_2}{\partial B}
    +a\frac{\partial v_1}{\partial A}\frac{\partial v_5}{\partial B}
    +c\frac{\partial v_3}{\partial A}\frac{\partial v_2}{\partial B}
    +\frac{\partial v_3}{\partial A}\frac{\partial v_5}{\partial B}
\]
With $a,b,c,d>0$, all partial derivatives are non-negative. Thus, $\det(\vec{J})$ is sign-definite and non-negative for all concentrations if $ac\geq bd$, whereas for $ac<bd$, the sign depends on the concentrations at which $\vec{J}$ is evaluated. 

The two line-graphs in this example are $L_1=C_1\cup C_2$ and $L_2=C_3$ with cycles $C_i$ as above, and thus $\mathcal{L}(G_I)=\{L_1,L_2\}$. Thus,
\begin{align*}
  \det(\vec{J}) &= \sum_{L\in\{L_1,L_2\}} \prod_{C\subseteq L} \sigGI(C)\,\gamma(C) \\
  &= (-1)^0\cdot \left(-a\frac{\partial v_1}{\partial A} - \frac{\partial v_3}{\partial A}\right)
            \cdot (-1)^0\cdot 
            \left(c\frac{\partial v_2}{\partial B} - \frac{\partial v_4}{\partial B}\right) \\
            &\quad\;+\; (-1)^1\cdot 
            \left(d\frac{\partial v_2}{\partial B}\right)\cdot\left(b\frac{\partial v_1}{\partial A}\right)\;.
\end{align*}
Because the edge-labels are sums of terms, the above condition for positivity of the determinant cannot be directly derived from the cycles alone.
\end{example}

\section{The Directed Species-Reaction Graph}
Analysis of qualitative dynamics via interaction graphs is considerably hampered by the fact that most networks contain sign-indefinite edges. Moreover, edge-labels are often sums of terms containing different rate-derivatives, making them hard to compare independently of a species concentration $\vec{x}$. These problems can all be addressed by exploiting the particular structure of chemical reaction systems which naturally leads to a bipartite graph with vertices for species and reactions. Our proposed \emph{directed species-reaction graph} directly relates to previous definitions of bipartite graphs and to the interaction graph.

\begin{definition}[directed species-reaction graph]
\label{def:dsr}
The \emph{directed species-reaction graph (DSR-graph)} $G=(V_S,V_R,E,\lambda)$ of a chemical reaction network is a bipartite, directed graph with edge-label function $\lambda$ given by the sets of
\begin{align*}
  \text{\emph{species vertices} }V_S &= \{S_1,\dots,S_n\} \\
  \text{\emph{reaction vertices} }V_R &= \{R_1,\dots,R_r\} \\
  \text{edges }E &= E_{SR}\cup E_{RS} \text{ consisting of}\\
  \text{\emph{rate edges} }E_{SR} &= \left\{(S_i,R_j)\in V_S\times V_R \,|\, \partial v_j/\partial x_i\not\equiv 0\right\}\text{ and}\\
    \text{\emph{stoichiometric edges} }E_{RS} &= \left\{(R_j,S_i) \in V_R\times V_S \,|\, N_{i,j}\not=0 \right\}
\end{align*}
and the
\begin{align*}
  \text{\emph{edge-label function} }\lambda((a,b)) &= \begin{cases}
    \frac{\partial v_j}{\partial x_i}, & \text{if } (a,b)=(S_i,R_j)\in E_{SR} \\
    N_{i,j}, &\text{if } (a,b)=(R_j,S_i)\in E_{RS} \\
    0, &\text{else.}
  \end{cases}     
\end{align*}
The restriction of $\lambda$ to $E_{SR}$ and $E_{RS}$ is denoted by $\lambda_{SR}$ (a \emph{rate label}) and $\lambda_{RS}$ (a \emph{stoichiometric label}), respectively, such that for a sub-graph $H$ of $G$,
\[
  \lambda_{RS}(H) := \left.\lambda\right|_{E_{RS}}(H) = 
    \prod_{(S_i,R_j)\in E_{RS}(H)} N_{ij} \in\Rset
\]  
and
\[
  \lambda_{SR}(H) := \left.\lambda\right|_{E_{SR}}(H) = \prod_{e\in E_{SR}(H)} \lambda(e)\;.
\]
Importantly, $\lambda_{SR}(H)(\vec{x})>0$ for all $\vec{x}>\vec{0}$. Either label remains undefined if the respective edge-set is empty.
\end{definition}
The generic DSR-graph of the network from Example~\ref{ex:small} is given in Figure~\ref{fig:dsr}. This graph has an intuitive interpretation: a reaction vertex represents the rate of the reaction, which is positively influenced only by its reactant species. A change in the rate on the other hand implies a positive change in the rate of the products, and a negative change in the rate of the reactants; this is reflected by the corresponding edges. For NAC rate laws, the rate-edge labels are positive functions, and the label of a stoichiometric edge is a positive constant if the species is a product, and a negative constant if it is a reactant of the respective reaction. In contrast to the interaction graph, sub-graphs of $G$ are therefore always sign-definite.

\begin{proposition}[sign-definite sub-graphs]
\label{prop:DSRsign-definite}
Let $H\subseteq G$ be any sub-graph of a DSR-graph $G$ such that $E_{RS}(H)\not=\emptyset$. Then,
\[
  \sign(\lambda(H)) \;\equiv\; \sign\left(\lambda_{RS}(H)\right)
\]
is independent of $\vec{x}>\vec{0}$. 
\end{proposition}

We briefly discuss some main differences of the DSR-graph to previous definitions of bipartite species-reaction graphs: The \emph{SR-graph} of~\cite{craciun2006} uses undirected edges and labels them by the \emph{complex} in which the species occurs. Species on the same side of a reaction form a \emph{c-pair}. All possible orientations of the graph are considered in proofs. In the DSR-graph, these information are encoded explicitly in the existence and direction of edges and two products do not form a c-pair. Instead of c-pairs, a similar undirected graph in~\cite{banaji2010} labels edges by $+1$ or $-1$ to the same effect. The graph proposed in~\cite{mincheva2007} uses directed edges, but does not contain edges from a reaction to its substrate. Directed edges from a substrate to its reaction can instead be traversed in opposite direction, while directed edges from a reaction to a product cannot, which also necessitates to allow semi-cycles and paths in a line-graph. In this graph, mass-action kinetics is also exploited by merging the corresponding factor from $\partial v_j/\partial x_i$ (which is the substrate molecularity of $S_i$ in $R_j$) with the stoichiometric label and using relative concentrations.

\begin{figure}[htbp]
\begin{center}
\includegraphics[width=\textwidth]{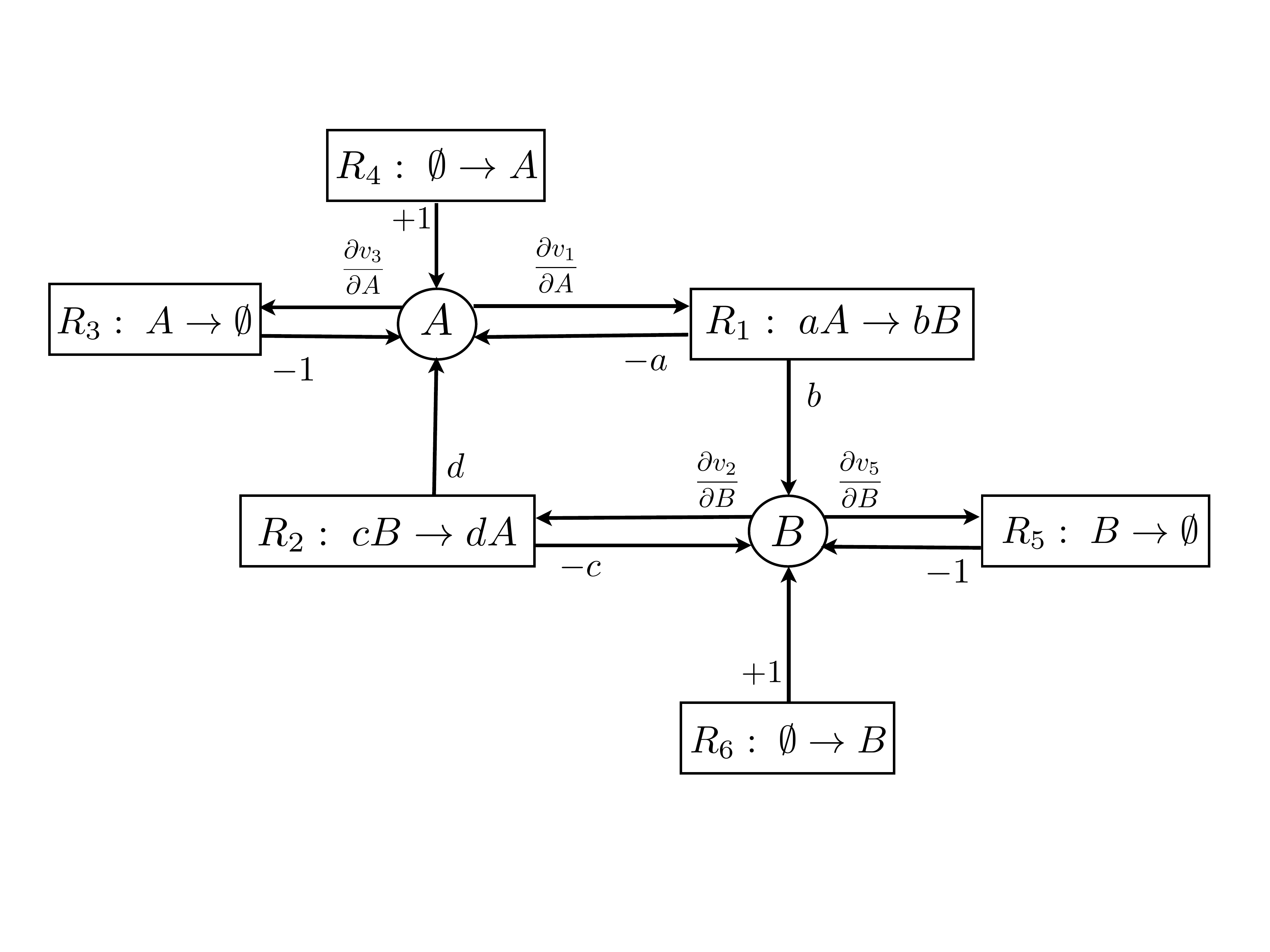}
\caption{Directed species-reaction graph for network of Example~\ref{ex:small}. Species vertices are given as circles, reaction vertices as rectangles. Production reactions $R_4, R_6$ are shown explicitly, but can be neglected for analyses. We again used $A,B$ instead of $x_A,x_B$ for readability.}
\label{fig:dsr}
\end{center}
\end{figure}

The proposed DSR-graph directly relates to the interaction graph of the same network. Most results on injectivity and bifurcations rely on the fact that simple paths and simple cycles in $G_I$ translate to simple paths and simple cycles in $G$, and we therefore emphasize this relation.

\begin{lemma}[relation of DSR- and interaction graph]
\label{lem:GIandG}
Let $G=(V_S,V_R,E,\lambda)$ be the DSR-graph of a chemical reaction network. The interaction graph $G_I=(V(G_I),E(G_I),\gamma)$ of that network is then found as:
\begin{align*}
  V(G_I) &= V_S \\
  E(G_I) &= \{(S,S') \in V_S\times V_S\,|\, \exists R\in V_R: (S,R),(R,S')\in E\} \\
  \gamma((S,S')) &= \sum_{R\in V_R} \lambda((S,R))\cdot\lambda((R,S'))\;.
\end{align*}
\end{lemma}
\begin{proof}
The only technicality is to see that if $(S_i,S_j)\in E(G_I)$, then there is at least one reaction $R_l$ such that both $(S_i,R_l)$ and $(R_l,S_j)$ are edges in $E$. This is because 
\[
  (S_i,S_j)\in E(G_I) \iff \frac{\partial f_j}{\partial x_i} \not\equiv 0
  \iff \exists l: N_{j,l} \frac{\partial v_l}{\partial x_i} \not\equiv 0\;.
\]
Indeed, $\frac{\partial f_j}{\partial x_i}  =\sum_{l} N_{j,l} \frac{\partial v_l}{\partial x_i}$, which allows reconstruction of $\gamma$ from $\lambda$.
\end{proof} 

An edge in $G_I$ may thus correspond to several 2-paths from $V_S$ to $V_S$ in $G$. Each summand of the edge-label in $G_{I}$ corresponds exactly to one of the labels of a 2-path in $G$.

\begin{example}
In the network of Example~\ref{ex:small}, consider the upper-left entry 
\[
  J_{1,1} \;=\; -a\frac{\partial v_{1}}{\partial x_{A}}-\frac{\partial v_{3}}{\partial x_{B}}
\]
in the Jacobian matrix. The corresponding edge $(A,\, A)$ in $G_{I}$ corresponds to the 2-paths $(A,R_{1},A)$ with label $-a\frac{\partial v_{1}}{\partial x_{A}}$ and $(A,R_{3},A)$ with label $-\frac{\partial v_{3}}{\partial x_{B}}$.
\end{example}

With $\vec{J}=\vec{N}\cdot(\partial \vec{v} / \partial \vec{x})$, the stoichiometric matrix is the incidence matrix describing stoichiometric edges and their labels in $G$, while $(\partial \vec{v} / \partial \vec{x})$ is the incidence matrix describing the rate edges and their labels. The $(n+r)\times (n+r)$ adjacency matrix of $G$ is
\begin{equation*}
  \vec{B}=\begin{pmatrix}
    \vec{0} & \vec{N} \\
    (\partial \vec{v} / \partial \vec{x}) & \vec{0}
  \end{pmatrix}\;.
\end{equation*}
All 2-paths from $V_S$ to $V_S$ are described by the upper-left $n\times n$ sub-matrix of $\vec{B}^2$, which is just $\vec{J}$. By extension, an edge, a simple cycle or a line-graph in $G_I$ typically correspond to several 2-paths, cycles or sub-graphs in $G$, respectively. This one-to-many mapping induces an equivalence relation on $G$.

\begin{definition}[equivalence; species-cycle; species-line-graph]
\label{lem:equiv-eCL}
Consider a reaction network with DSR-graph $G=(V_S,V_R,E,\lambda)$ and interaction graph $G_I=(V(G_I),E(G_I),\gamma)$. Let $e=(S,S')\in E(G_I)$ be an edge and denote by
\[
  \langle e \rangle := \{((S,R),(R,S'))\in E\times E\,|\, R\in V_R\}
\]
the set of all corresponding 2-paths. Two 2-paths $p,p'$ are \emph{$G_I$-equivalent} if $p,p'\in \langle e \rangle$.

The equivalence relation is extended to cycles $C=(e_1,\dots,e_q)$ and line-graphs $L=C_1+\cdots +C_q$ of $G_I$ by 
\begin{align*}
  \langle C \rangle &:= \{(p_1,\dots,p_q)\,|\, p_i\in \langle e_i \rangle\} \\
  \langle L \rangle &:= \{D_1+\cdots +D_q\,|\,D_i\in \langle C_i \rangle\}\;.
\end{align*} 
An element in $\langle C \rangle$ or $\langle L \rangle$ is called a \emph{species-cycle} or \emph{species-line-graph} in $G$, respectively. Each species-line-graph is a set of simple disjoint species-cycles covering each species-vertex exactly once. We again denote the set of all species-line-graphs in $G$ by
\begin{equation*}
  \mathcal{L}(G) := \bigcup_{L\in\mathcal{L}(G_I)} \langle L \rangle\;.
\end{equation*}
\end{definition}

\begin{example}
\label{ex:mm}
Consider a Michaelis-Menten type mechanism, given by reactions
\begin{align*}
  R_{1,2}:& E+S\rightleftarrows ES \\
  R_3:& ES\to E+P
\end{align*}
The DSR-graph for this mechanism is given in Figure~\ref{fig:dsrmm}. The edge $e=(ES,E)\in E(G_I)$ corresponds to the equivalence class $\langle e \rangle=\{(ES,R_2,E), (ES,R_3,E)\}$ in $G$. The cycle $C=(ES,S,ES)$ of $G_I$ thus corresponds to two species-cycles in $G$, using either two-path from $\langle e \rangle$ together with the two-path $(E,R_1,ES)$:
\[
  \langle C \rangle  =  \left\{(ES,R_2,E),\,(E,R_1,ES),\,(ES,R_3,E,ES) \right\} \;.
\]

\begin{figure}[htbp]
\begin{center}
\includegraphics[width=\textwidth]{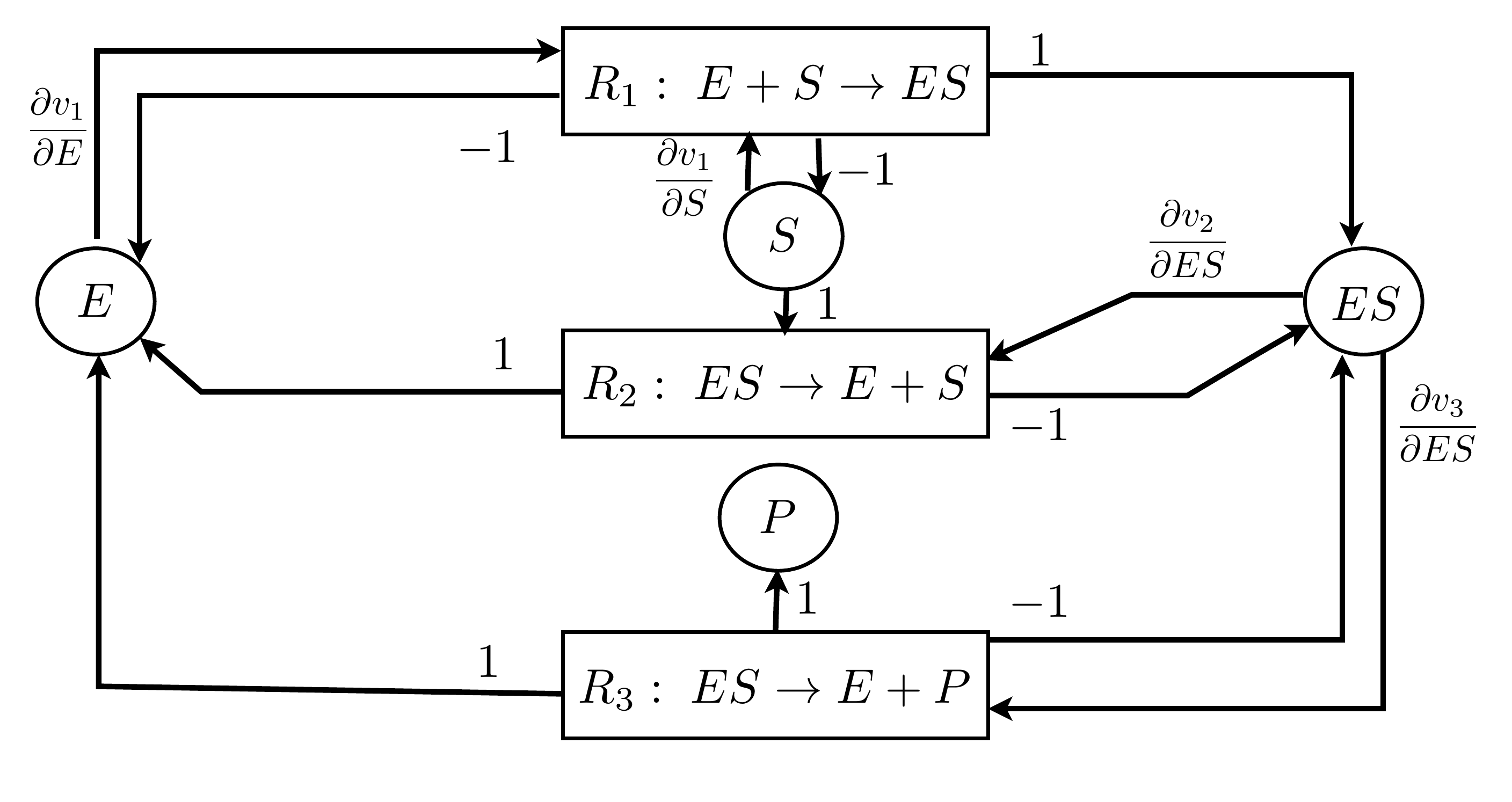}
\caption{DSR-graph of Michaelis-Menten type mechanism of Example~\ref{ex:mm}.}
\label{fig:dsrmm}
\end{center}
\end{figure}
\end{example}

The label-function $\gamma$ of $G_I$ can be reconstructed from the label-function $\lambda$ of $G$ by summing over elements of the corresponding equivalence class. For cycles and line-graphs, this allows either a product-of-sums or a sum-of-products representation of the label.
\begin{lemma}[label-function]
\label{lem:gamma-reconstruct}
Let $e$ be an edge, $C$ be a cycle, and $L$ be a line-graph of an interaction graph $G_I$ and let $\lambda$ be the edge-label function of the corresponding DSR-graph $G$. Then,
\begin{align*}
  \gamma(e) &= \sum_{p\in \langle e \rangle} \lambda(p)\\
  \gamma(C) &= \sum_{C'\in \langle C \rangle} \lambda(C') = \prod_{e\in C}\sum_{p\in \langle e \rangle}\lambda( p) = 
    \sum_{C'\in \langle C\rangle}\prod_{e\in C'}\lambda(e)\\
  \gamma(L) &= \sum_{L'\in \langle L \rangle} \lambda(L') =\prod_{C\subseteq L}\sum_{C'\in \langle C \rangle}\lambda(C') =
    \sum_{L'\in \langle L \rangle}\prod_{C'\subseteq L'}\lambda(C')\;.
\end{align*}
\end{lemma}

\section{Non-vanishing Determinants and DSR-Graph}
We now turn our attention to the expansion of $\det(\vec{J})$ (and consequently $\det(\vec{J}_\alpha)$) in terms of the DSR-graph. We are particularly interested in conditions that guarantee that the determinant does not vanish for any positive states $\vec{x}>\vec{0}$. 

\begin{lemma}[signum of sub-graph]
Let $H\subseteq G$ be a sub-graph of G and define the \emph{signum} of $H$ as
\[
  \sigma(H):=(-1)^{\epsilon(H)}
\]
with $\epsilon(H)$ the number of cycles in $H$ with even number of species-vertices. Then,
\[
  \sigGI(C) \;=\; \sigma(D)
\]
for any cycle $C\subseteq G_I$ and $D\in \langle C \rangle$.
\end{lemma}
\begin{proof}
A species-cycle $D$ in $G$ with $k$ species-vertices corresponds to a cycle $C$ in $G_{I}$ of length $k$. Thus, $\omega( C)=(-1)^k=\sigma(D)$.
\end{proof}

A determinant expansion purely in terms of a DSR-graph is now easily found.
\begin{lemma}[determinant expansion by DSR-graph]
\label{lem:detexpansionGSR}
Consider a chemical reaction network with Jacobian $\vec{J}$ and DSR-graph $G$. Then,
\[
  \det(\vec{J}) = \sum_{L\in\mathcal{L}(G)} \sigma(L)\, \lambda(L) \;.
\]
\end{lemma}
\begin{proof}
Applying the definition of $\mathcal{L}(G)$ from Lemma~\ref{lem:equiv-eCL}, we reduce the expression to the one found for $G_I$ in Lemma~\ref{lem:detexpansionGI}:
\[
  \sum_{L\in\mathcal{L}(G)} \sigma(L)\, \lambda(L) 
  = \sum_{L\in\mathcal{L}(G_{I})}\sum_{L'\in\langle L \rangle} \sigGI(L')\, \gamma(L') 
  = \sum_{L'\in\mathcal{L}(G_I)} \sigGI(L')\, \gamma(L')\;,
\]  
where $\sigGI$ and $\gamma$ are again the signum and label function in $G_I$. 
\end{proof}

\begin{example}
\label{ex:lgraphs}
Consider the DSR-graph of Example~\ref{ex:small}, given in Figure~\ref{fig:dsr}. Its species-line-graphs are
\begin{align*}
  L_{1}:& (A,R_{1},B,R_{2},A) \\
  L_{2}:& (A,R_{1},A) \cup (B,R_{2},B) \\
  L_{3}:& (A,R_{1},A) \cup (B,R_{5},B) \\
  L_{4}:& (A,R_{3},A) \cup (B,R_{2},B) \\  
  L_{5}:& (A,R_{3},A) \cup (B,R_{5},B)
\end{align*}
corresponding directly to the five expansion terms
\[
  \det(\vec{J}) \;=\; \underbrace{-\;bd\frac{\partial v_1}{\partial x_A}\frac{\partial v_2}{\partial x_B}}_{L_{1}}  
    \underbrace{+ac\frac{\partial v_1}{\partial x_A}\frac{\partial v_2}{\partial x_B}}_{L_{2}}
    \underbrace{+a\frac{\partial v_1}{\partial x_A}\frac{\partial v_5}{\partial x_B}}_{L_{3}}
    \underbrace{+c\frac{\partial v_3}{\partial x_A}\frac{\partial v_2}{\partial x_B}}_{L_{4}}
    \underbrace{+\frac{\partial v_3}{\partial x_A}\frac{\partial v_5}{\partial x_B}}_{L_{5}}\;.
\]   

\end{example}
  
In contrast to the determinant expansion from $G_I$, the sum-of-products representation of labels of $G$ yields a direct correspondence of species-line-graphs and expansion terms. Terms have the same partial derivatives if their species-line-graphs have identical substrate-reaction edges. This observation motivates to identify \emph{compatible} line-graphs in $G$ and determine their overall contribution to the expansion from their stoichiometric labels.

\begin{proposition}[compatibility]
Let $H,H'$ be two sub-graphs of $G$. The relation
\[
  H\sim H' \iff H,H'\text{ \emph{compatible} } :\iff E_{SR}(H) = E_{SR}(H')
\]
defines an equivalence relation. We write
\[
  [H] := \{H\subseteq G \,|\, H'\sim H\}
\]
for the equivalence class of a sub-graph $H$. In particular, $\sim$ partitions $\mathcal{L}(G)$ into equivalence classes in the quotient set $\mathcal{L}(G)/\sim$
\end{proposition}
\begin{proof}
Reflexivity, symmetry, and transitivity of $\sim$ are obvious.
\end{proof}

As an example, $L_{1}$ and $L_{2}$ in Example~\ref{ex:lgraphs} are the only non-trivially compatible species-line-graphs. The notion of compatibility suggests a strategy to determine if $\det(\vec{J})$ vanishes by summing over each individual compatibility class of $\mathcal{L}(G)$. If all classes are either non-negative or non-positive, the sign of the determinant is independent of the state $\vec{x}$. 

\begin{definition}[stoichiometric term of equivalence class]
Let $G$ be a DSR-graph and consider a compatibility class $[L]\in\mathcal{L}(G)/\sim$. The term
\[
  \Lambda([L]) := \sum_{L'\in [L]} \sigma(L')\lambda_{RS}(L')
\]
is called the \emph{stoichiometric term} of $[L]$. It is a constant independent of $\vec{x}$.
\end{definition}

The stoichiometric term can be computed from the stoichiometric matrix alone.
\begin{lemma}[computing stoichiometric terms]
\label{lem:computeW}
Fix a species-line-graph $L\in\mathcal{L}(G)$. Let $r_{j}$ be the index of the unique reaction with substrate $S_{j}$ in $L$. Define the $n\times n$ matrix $\vec{W}_{L}$ by
\[
  W_{i,j} = \begin{cases} 
    1,&\text{if } (S_j,R_{r_{j}})\in E_{SR}(L) \\
    0,&\text{else}\;,
    \end{cases}
\]    
and let $\vec{N}_{L}$ be the $n\times n$ stoichiometric matrix with columns not in $r_{1},\dots,r_{n}$ removed.
Then,
\[
  \Lambda([L]) = \det\left(\vec{N}_{L}\cdot\vec{W}_{L}\right) = \pm \det\left(\vec{N}_{L}\right)
\]
with the sign determined uniquely by $\vec{W}_{L}$.
\end{lemma}
\begin{proof}
The matrix $\vec{N}_{L}\cdot\vec{W}_{L}$ corresponds to a graph with only rate edges corresponding to substrate-reaction pairs of $[L]$ and all other reaction vertices removed; its determinant sums the contributions of all remaining species-line-graphs in that graph. Moreover, $\vec{W}_{L}$ is simply a permutation matrix whose determinant is thus $\pm 1$.
\end{proof}

\begin{theorem}[determinant expansion by compatibility classes]
\label{thm:detsign}
Fix any NAC reaction system and let $\vec{J}$ be its Jacobian matrix and $G$ its DSR-graph. Then,
\[
  \det(\vec{J}) = \sum_{[L]\in\mathcal{L}(G)/\sim} \;\Lambda([L])\cdot\lambda_{SR}(L)\;.
\]
The determinant is non-negative for all $\vec{x}>\vec{0}$ if
\[
  \Lambda([L])\geq 0 \text{ for all } [L]\in\mathcal{L}(G)/\sim
\]
and positive if in addition $\Lambda([L])> 0$ for at least one compatibility class. Similar conditions hold for non-positivity (negativity) of the determinant.
\end{theorem}
\begin{proof}
Because $\sim$ is an equivalence relation on $\mathcal{L}(G)$, we can partition the sum of Lemma~\ref{lem:detexpansionGSR} by each equivalence class. Note that the rate label is identical and non-negative for all members of a class. If all stoichiometric terms are non-negative, then so is the sum.
\end{proof}

As a consequence, we can restrict our attention to finding conditions that establish non-negativity (or non-positivity) of the stoichiometric term for each compatibility class. For reasons of convenience that will become apparent shortly, we will focus our attention on $G(\vec{-J})$, which is found by inverting signs of stoichiometric edges. With $\det(\vec{J})=(-1)^n\det(\vec{-J})$, the determinant expansion in $G(\vec{J})$ is obviously non-zero if and only if the expansion in $G(\vec{-J})$ is. We start by giving a sufficient condition to find a positive expansion term in open networks.

\begin{lemma}[existence of positive term in open networks]
\label{lem:Lout}
Consider a reaction network and assume that there is an inflow reaction $\emptyset\to S_i$ and an outflow reaction $R_i: S_i\to\emptyset$ for each species $S_i$, $1\leq i\leq n$. Consider the DSR-graph $G(-\vec{J})$ and let $L_{\text{out}}$ be the species-line-graph
\[
  L_{\text{out}} := \bigcup_{i=1}^n\;(S_i,R_i,S_i) \;,
\]
of 2-paths from each species to itself via its outflow reaction. Then,
\begin{enumerate}
\item $[L_{\text{out}}] = \{L_{\text{out}}\} $
\item $\Lambda([L_{\text{out}}]) > 0$
\item $\Lambda([L_{\text{out}}]) \cdot \lambda_{SR}(L_{\text{out}}) > 0$ for all $\vec{x}>\vec{0}$
\end{enumerate}
\end{lemma}
\begin{proof}
First note that $S_i$ is a substrate to $R_i$, so $((S_i,R_i),\,(R_i,S_i))\in E\times E$ for each species $S_i$. Thus, $L_{\text{out}}$ is indeed a species-line-graph. Since there is no other edge out of an outflow reaction, no other species-line-graph can use the same substrate-reaction pairs and thus the compatibility class has a single element. Each cycle in $L_{\text{out}}$ has exactly one species-vertex and thus $\sigma(L_{\text{out}})=+1$. Further, its stoichiometric edge is positive in $G(-\vec{J})$ and thus the stoichiometric term of $[L_{\text{out}}]$ is also. For positive concentrations, the outflow rate changes are positive, proving the last claim. 
\end{proof}

Next, we show that species-line-graphs that either contain both the forward- and backward reaction of a reversible reaction always lead to a zero stoichiometric term for their compatibility class. This result also gives a retrospective justification for splitting reversible reactions.
\begin{lemma}[zero-contribution of reversible reaction]
\label{lem:fwbw}
Let $L\in\mathcal{L}(G)$ be a species-line-graph. Let $R_f, R_b$ be the forward and backward reaction of a reversible reaction and assume $R_f,R_b\in V_R(L)$. Then,
\[
  \Lambda([L]) \equiv 0\;.
\]
\end{lemma}
\begin{proof}
The construction in this proof is illustrated in Fig.~\ref{fig:proofillustration} (left). Denote by $S_f, S_r\in V_S(L)$ the species with $(R_f, S_f), (R_b,S_b)\in E_{RS}(L)$. Because $R_f, R_b$ constitute one reversible reaction, the edges $(R_f, S_b), (R_b, S_f)$ exist in the graph $G$. Construct a sub-graph $L'$ by replacing $(R_f, S_f), (R_b,S_b)$ by $(R_f, S_b), (R_b, S_f)$. If $R_f, R_b$ are contained in one species-cycle $C_3\in L$, they are now contained in two different cycles $C_1, C_2\in L'$ (or vice-versa). Thus, $L'$ is a species-line-graph and $L'\sim L$. Moreover, $\lambda((R_f, S_f))=-\lambda((R_b, S_f))$ and $\lambda((R_b, S_b))=-\lambda((R_f, S_b))$ and thus $\lambda(L)=\lambda(L')$. However, let $n_i:=|E(C_i)|$, then
\[
  \sigma(C_1)\sigma(C_2)=(-1)^{n_1-1}\,(-1)^{n_2-1}=(-1)^{n_3-2}=-\sigma(C_3)
\]
and thus $\sigma(L')\lambda(L')=-\sigma(L)\lambda(L)$. This construction gives a bijection between species-line-graphs $L$ and $L'$ and thus $\Lambda([L]) \equiv 0$.
\end{proof}

\begin{figure}[htbp]
\begin{center}
\includegraphics[width=\textwidth]{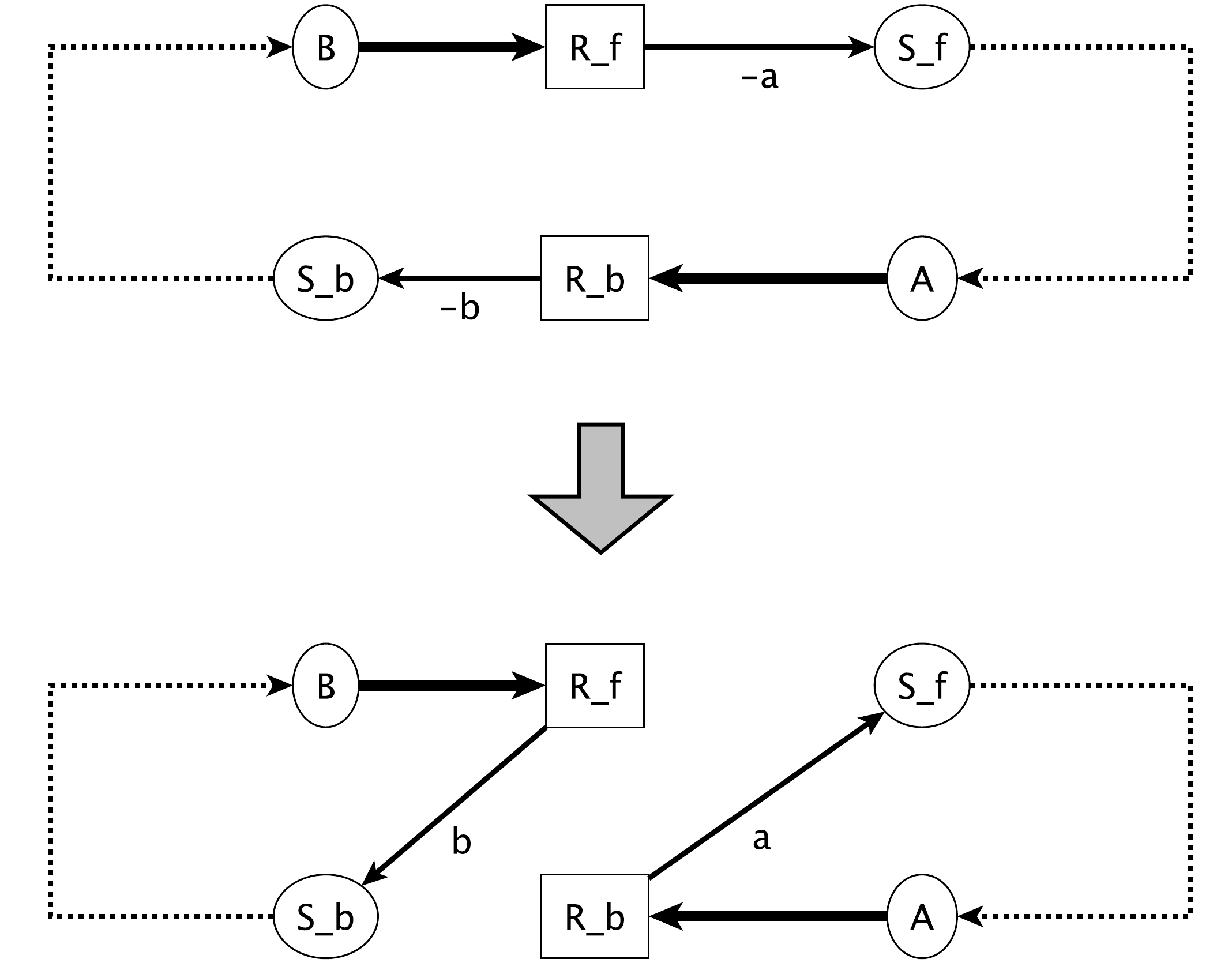}
\caption{Illustration of construction for Lemma~\ref{lem:fwbw}. A cycle containing the forward- and backward reaction can be split into two cycles, leading to a compatible line-graph with same absolute label, but opposite sign. Edge-labels denote stoichometric label. Dotted eges denote arbitrary paths through the graph, bold edges denote fixed substrate-reaction edges defining the compatibility class.}
\label{fig:proofillustration}
\end{center}
\end{figure}

\begin{figure}[htbp]
\begin{center}
\includegraphics[width=\textwidth]{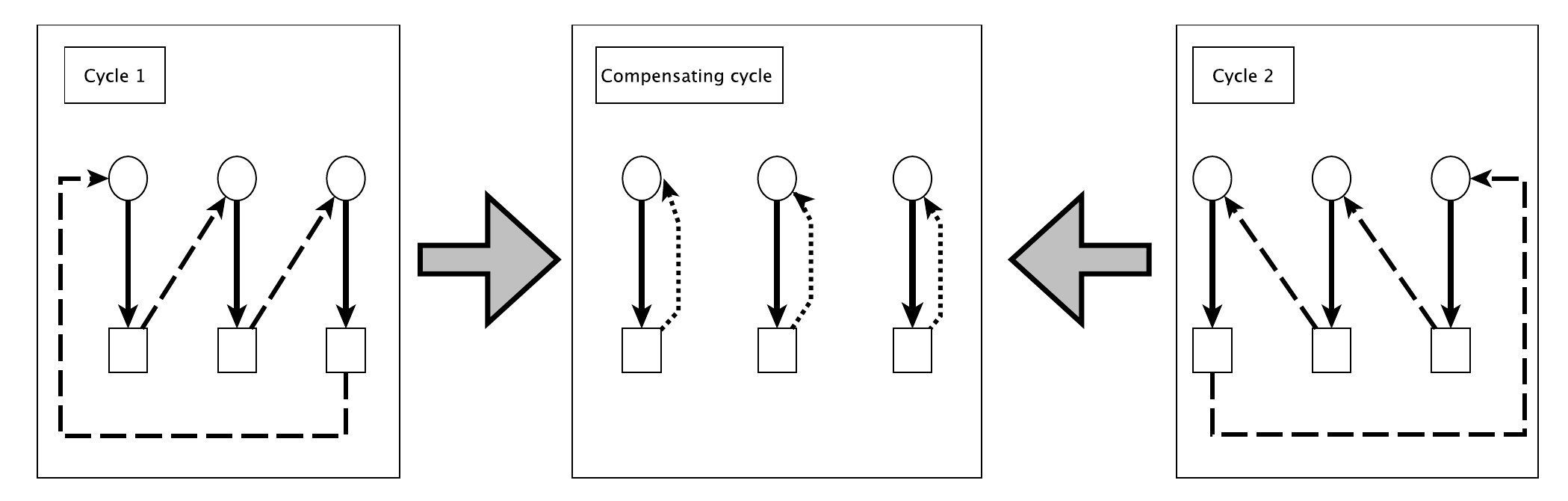}
\caption{Illustration of compensating cycles. For a cycle Cycle 1 ($C$), a compatible compensating cycle (center, $L_C$) always exists, whose stoichiometries might dominate. A second compatible cycle Cycle 2 may however yield the same compensating cycle. Round (rectangular) nodes denote species (reaction) vertices.}
\label{fig:compcycle}
\end{center}
\end{figure}

\begin{lemma}[existence of positive compatible species-line-graphs]
\label{lem:existLD}
Let $G(-\vec{J})$ be a DSR-graph and consider a species-line-graph $L$ containing a cycle $C$. Construct a sub-graph $L_C\subseteq G$ by replacing each 2-path $((S,R),(R,S'))$ in $C$ by $((S,R),(R,S))$. Then,
\begin{enumerate}
  \item $\sigma(L_C)\lambda(L_C) > 0$ for all $\vec{x}>\vec{0}$ 
  \item $L_C \sim C$
  \item $(L\backslash C)\cup L_C \sim L$.
\end{enumerate} 
\end{lemma}
\begin{proof}
The construction is illustrated in Fig.~\ref{fig:compcycle}. With $(S,R)\in E( C)$, species $S$ is a reactant of $R$. Then, also $(R,S)\in E(G)$ and thus $L_C$ is a proper sub-graph of $G$ with exactly one odd-length cycle $D_i$ from each species-vertex $S_i\in V_S( C)$ to itself. Thus, $\sigma(D_i)=+1$ and $\lambda(D_i)>0$ for all $\vec{x}>\vec{0}$ and $i=1,\dots,n$. These cycles use the same substrate-reaction pairs as $C$ and thus $C\sim L_C$. They also cover each species-vertex in $C$ exactly once, so $L \sim (L\backslash C)\cup L_C$. As a special case, $L_C=C$ if $|V_S( C)|=1$.
\end{proof}

The sign of the contribution of a species-line-graph depends on its constituent species-cycles. Since each sub-graph of a DSR-graph is sign-definite, we can give simple conditions for a species-cycle to be positive or negative by determining the number of substrate-to-substrate 2-paths in the cycle.
\begin{lemma}[condition for sign of species-cycles]
Consider a species-cycle $C$ in a DSR-graph $G(\vec{-J})$. Let $s$ be the number of \emph{substrate-pairs}, that is, of 2-paths $(S,R,S')$ in $C$, such that $S,S'$ are both a substrate of $R$. We call $C$ a \emph{p-cycle} (an \emph{n-cycle}) if $\sigma( C)\lambda( C) > 0$) (resp.\ $<0$). Then,
\begin{align*}
  C \text{ is p-cycle} &\iff s \text{ is even} \\
  C \text{ is n-cycle} &\iff s \text{ is odd}\;.
\end{align*}   
\end{lemma}
\begin{proof}
We consider the four possible combinations of $V_S( C)$ even/odd and $s$ even/odd. For $V_S( C)$ even and $s$ odd, the number of negative 2-paths is also odd, thus $\sigma( C)=+1$ and $\lambda( C)<0$, as $E( C)$ contains an odd number of negative stoichiometric edges. Thus, the overall contribution of $C$ is negative. The other three cases follow the same reasoning. 
\end{proof}
In~\cite{craciun2006,banaji2010} p-cycles (n-cycles) were called e-cycles (o-cycles).  As a consequence of Lemma~\ref{lem:existLD}, we can give a simple condition when a negative contribution to the determinant expansion is cancelled (see also~\cite{craciun2006,banaji2010}).
\begin{proposition}[dominating term]
\label{prop:cancel}
Let $G(-\vec{J})$ be a DSR-graph and consider a species-line-graph $L$ containing an n-cycle $C$. Let $L'=(L\backslash C)\cup L_C$. Then,
\[
  \sigma(L)\lambda(L) + \sigma(L')\lambda(L') \geq 0 \iff \lambda_{RS}(L_C)\geq|\lambda_{RS}( C)|\;.
\]
We then say that $L_C$ \emph{dominates} $C$. We call $C$ a \emph{stoichiometric} or \emph{s-cycle}, if $\lambda_{RS}( L_C)=|\lambda_{RS}( C)|$.   
\end{proposition}

The previous Lemma~\ref{lem:existLD} and Proposition~\ref{prop:cancel} hold the key to determine, for all $\vec{x}>\vec{0}$, if a determinant vanishes or not. Clearly, for a species-line-graph $L$ to give a negative contribution to the stoichiometric term $\Lambda([L])$ of its compatibility class, it contains an odd number of n-cycles. Replacing one of these n-cycles $C$ by $L_C$ leads to a new compatible species-line-graph with positive contribution that compensates the negative. However, there might be a second species-line-graph in the same class, also with negative contribution that contains another negative cycle that leads to the same compensating species-line-graph. Thus, the compensating species-line-graph needs to dominate \emph{the sum} of all these contributions. Next, we give conditions when such situation arises and provide a simple sufficient condition that excludes it.

\begin{theorem}[species-reaction intersection of species-cycles]
Consider a species-line-graphs $L$ in $G(\vec{-J})$ and assume that in every compatible species-line-graph, each n-cycle $C$ is dominated by $L_C$. Further assume that
\[
  \Lambda([L]) < 0\;.
\]  
Then, there are $L_1,L_2\in [L]$ and two non-disjoint n-cycles $C_1\subseteq L_1$, $C_2\subseteq L_2$ such that all paths in $C_1\cap C_2$ start in $V_S$ and end in $V_R$.
\end{theorem}
\begin{proof}
An illustration for this proof is given in Fig.~\ref{fig:srintersect}. First note that the non-empty intersection of two cycles is always a collection of paths. Let $P$ be one of the paths in the intersection of $C_1,C_2$. If $P$ ends in $V_S$, there are two substrate-reaction pairs involving the same substrate species. Hence, $L_1,L_2$ cannot be compatible. If $P$ begins in $V_R$, there are two different species-reaction vertices leading into it, one from $C_1$, one from $C_2$. These edges are contained in all compatible sub-graphs leading to non-simple cycles. The sub-graphs are hence not species-line-graphs. The only remaining case is a path from $V_S$ to $V_R$ which yields both unique substrate-reaction pairs and (potentially) the same number of reactions and species.
\end{proof}

\begin{figure}[htbp]
\begin{center}
\includegraphics[width=\textwidth]{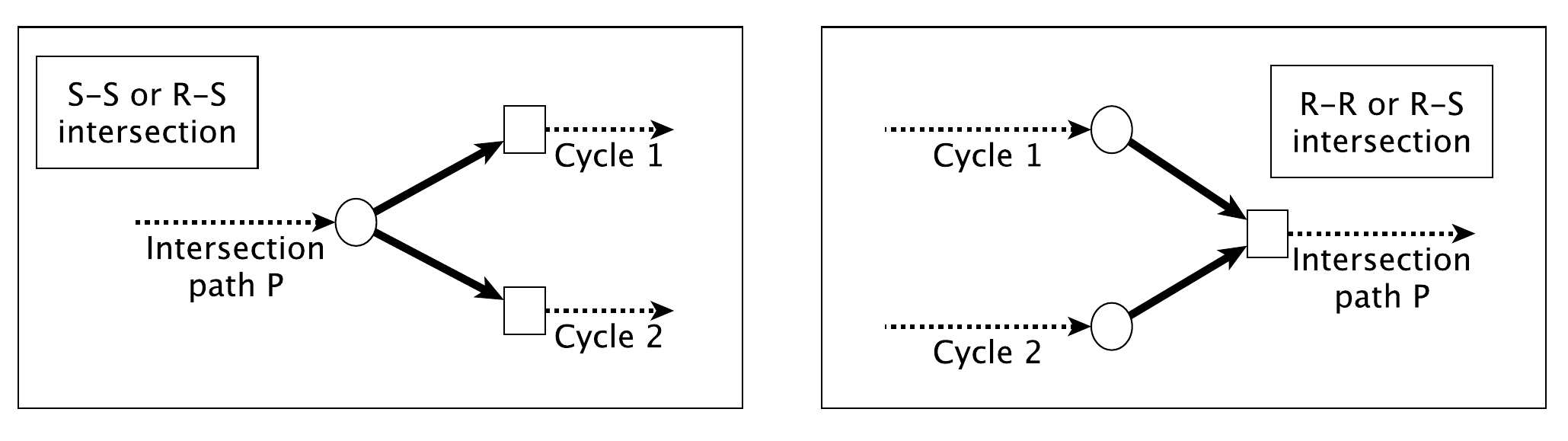}
\caption{Illustration of non-feasible intersection of cycles in compatible line-graphs. Left: Both possible intersections ending in $V_S$ require two substrate-reaction pairs with the same species and cannot occur in compatible sub-graphs. Right: An intersection from $V_R$ to $V_S$ uses a reaction vertex twice and leads to a zero overall contribution. Round (rectangular) nodes denote species (reaction) vertices. Dotted edges denote arbitrary paths through the graph, bold edges denote substrate-reaction pairs.}
\label{fig:srintersect}
\end{center}
\end{figure}

As a corollary of the theorem, we get a simple condition of the cycle-structure of a DSR-graph that allows to test if the determinant of the network does not vanishes anywhere. An equivalent condition was first formulated in~\cite{craciun2006}. 
\begin{corollary}[necessary condition for positive determinant]
\label{cor:neccond}
Consider a DSR-graph $G(\vec{-J})$ of an NAC network. The determinant $\det(\vec{-J})$ is positive if
\begin{enumerate}
\item there is at least one positive stoichiometric term
\item every cycle $C$ in $G$ is either a p-cycle or dominated by $L_C$
\item no two n-cycles have an intersection from $V_S$ to $V_R$
\end{enumerate}
\end{corollary}
The existence of a positive term is guaranteed for open networks by Lemma~\ref{lem:Lout}.	

\section{Extensions}
A possible extensions to the idea of decomposing the determinant expansion using equivalence classes can be achieved if a partial order on the rate derivatives $\partial v_i/\partial x_j$ can be established. Since $\lambda_{SR}(L)$ is a product of rate derivatives, this order induces a partial order on some compatibility classes $L,L'$ such that $\lambda_{SR}(L)\geq\lambda_{SR}(L')$ for all $\vec{x}>\vec{0}$. This order of classes implies consequently implies
\[
  \Lambda([L]) > \Lambda([L']) \implies \Lambda([L])\cdot\lambda_{SR}(L) > \Lambda([L'])\cdot\lambda_{SR}(L')\;,
\]
allowing to compare two equivalence classes purely by their stoichiometric term. A negative contribution of $[L']$ might thus be compensated by a larger positive one of $[L]$, independently of $\vec{x}$. 

For mass-action kinetics, such a partial order is given for all species-line-graphs that use the same set of reaction vertices. First, note that the rate derivatives for mass-action rate laws are 
\[
  \frac{\partial v_i}{\partial x_j} (\vec{x}^*) = -N_{j,i}\cdot v_i(\vec{x}^*)\cdot \frac{1}{x_j^*}\;.
\] 
Consider now two species-line-graphs $L,L'$ in $G(\vec{-J})$ with $V_R(L)=V_R(L')$. Then,
\begin{align*}
  \lambda_{SR}(L) &= \prod_{(S_j,R_i)\in E(L)} -N_{j,i} \cdot v_i(\vec{x}^*)\cdot \frac{1}{x_j^*} \\
  \lambda_{SR}(L') &= \prod_{(S_j,R_i)\in E(L')} -N_{j,i} \cdot v_i(\vec{x}^*)\cdot \frac{1}{x_j^*}\;,
\end{align*}
and thus
\[
  \frac{\lambda_{SR}(L)}{\lambda_{SR}(L')} \;=\; \frac{\prod_{(S_j,R_i)\in E(L)} -N_{j,i}}{\prod_{(S_j,R_i)\in E(L')} -N_{j,i}}\text{ independently of }\vec{x}>\vec{0}\;.
\]
This allows to combine the results for all compatibility classes with the same reaction-vertices, a strategy employed in~\cite{craciun2006,mincheva2007}: Let $V_R^*$ be the specific set of $n$ reactions, and let $G^*\subseteq G(\vec{-J})$ be the DSR-graph with $V_R(G^*)=V_R^*$. Then, all species-line-graphs in $G^*$ give a combined non-negative contribution to the determinant if
\begin{equation}
\label{eq:fragment}
  \sum_{[L]\in\mathcal{L}(G^*)/\sim} \Lambda([L])\cdot\left(\prod_{(S_i,R_j)\in E(L)} -N_{i,j}\right) \geq 0\;.
\end{equation}
We remark that this overall term can also be computed similar to Lemma~\ref{lem:computeW} by replacing $\vec{W}_{L}$ by a suitable $n\times n$ matrix extracted from $\vec{N}^{-}$. A term~(\ref{eq:fragment}) is called a \emph{fragment} (of size $n$) in~\cite{mincheva2007} and a \emph{critical fragment}, if it is negative. In that publication, the question of the relation between critical fragments and conditions on Corollary~ \ref{cor:neccond} was raised. Since the products $\prod (-N_{i,j})$ are all positive, Corollary~\ref{cor:neccond} gives a sufficient condition to exclude a critical fragment, as it establishes non-negativity for each summand in~(\ref{eq:fragment}).

\section{Discussion}
The particular structure of dynamic chemical reaction network models allows to derive conditions to show or exclude various qualitative dynamic behavior. These conditions enable a first analysis and model selection independently of numerical values for rate constants and for all members of a large class of rate laws. 

In this paper, we brought forward a new definition of a bipartite species-reaction graph, termed DSR-graph. In contrast to previous definitions, all relevant features of cycles, such as feasible directions to traverse edges and substrate/product relations of species and reactions are directly encoded in the graph. Our DSR-graph contains previous definitions as special cases. We elucidated the direct connection of the DSR-graph to the systems' interaction graph and demonstrated how cycle features can be mapped by simple equivalences of edges and 2-paths. For calculating terms of the determinant expansions of the Jacobian matrix, both graphs yield structurally similar formulas, but the DSR-graph allows a more fine-grained analysis of the terms. As a new result for bipartite graphs of chemical reaction systems, we proposed a simple equivalence relation on the species-line-graphs of the DSR-graph that allows to collect comparable terms in the expansion and subsequently enabled simpler and more direct proofs of conditions for the non-vanishing of principal minors of the Jacobian matrix. We finally addressed a question raised in~\cite{mincheva2007} on the relation of their conditions to the ones developed by Craciun et al in~\cite{craciun2006}.

\subsection*{Acknowledgements}
We would like to thank Irene Otero Muras,  Markus Beat D{\"u}rr and J{\"o}rg Stelling for fruitful discussions. Financial support from the EU FP7 project UNICELLSYS is gratefully acknowledged.

\bibliographystyle{plain}
\bibliography{dsr-paper}

\begin{thebibliography}{10}

\bibitem{angeli2010}
David Angeli, Patrick~De Leenheer, and Eduardo Sontag.
\newblock Graph-theoretic characterizations of monotonicity of chemical
  networks in reaction coordinates.
\newblock {\em J. Math. Biol.}, 61(4):581--616, Oct 2010.

\bibitem{banaji2010}
M.~Banaji.
\newblock Cycle structure in {SR} and {DSR} graphs: implications for multiple
  equilibria and stable oscillation in chemical reaction networks.
\newblock arXiv:1005.5472v2, 2010.

\bibitem{banaji2008}
Murad Banaji and Gheorghe Craciun.
\newblock Graph-theoretic criteria for injectivity and unique equilibria in
  general chemical reaction systems.
\newblock {\em Adv Appl Math.}, 44:168--184, 2010.

\bibitem{clarke1980}
Bruce~L. Clarke.
\newblock Stabiliy of complex reaction networks.
\newblock {\em Adv. Chem. Phys.}, 43:1--215, 1980.

\bibitem{craciun2006}
Gheorghe Craciun, Yangzhong Tang, and Martin Feinberg.
\newblock Understanding bistability in complex enzyme-driven reaction networks.
\newblock {\em Proc. Natl. Acad. Sci. USA}, 103(23):8697--8702, Jun 2006.

\bibitem{craciun2005}
Gheorghie Craciun and Martin Feinberg.
\newblock Multiple equilibria in complex chemical reaction networks: {I}. {The}
  injectivity property.
\newblock {\em SIAM J. Appl. Math.}, 65(5):1526--1546, May 2005.

\bibitem{craciun2006a}
Gheorghie Craciun and Martin Feinberg.
\newblock Multiple equilibria in complex chemical reaction networks: {II}.
  {The} species-reaction graph.
\newblock {\em SIAM J. Appl. Math.}, 66(4):1321--1338, Mar 2006.

\bibitem{domijan2011}
Mirela Domijan and Elisabeth P{\'e}cou.
\newblock The interaction graph structure of mass-action reaction networks.
\newblock {\em J. Math. Biol.}, Aug 2011.

\bibitem{feinberg1995a}
Martin Feinberg.
\newblock The existence and uniqueness of steady states for a class of chemical
  reaction networks.
\newblock {\em Arch. Rational Mech. Anal.}, 132:311--370, Aug 1995.

\bibitem{feinberg1995}
Martin Feinberg.
\newblock Multiple steady states for chemical reaction networks of deficiency
  one.
\newblock {\em Arch. Rational Mech. Anal.}, 132:371--406, Aug 1995.

\bibitem{gouze1998}
Jean-Luc Gouze.
\newblock Positive and negative circuits in dynamical systems.
\newblock {\em J. Biol. Syst,}, 6:11--15, 1998.

\bibitem{harary1962}
Frank Harary.
\newblock The determinant of the adjacency matrix of a graph.
\newblock {\em SIAM Review}, 4(3):202--210, July 1962.

\bibitem{helton2010}
J.~Helton, Vitaly Katsnelson, and Igor Klep.
\newblock Sign patterns for chemical reaction networks.
\newblock {\em J. Math. Chem.}, 47(1):403--429, January 2010.

\bibitem{hirsch2005}
M~W Hirsch and Hal Smith.
\newblock Monotone maps: a review.
\newblock {\em Journal of Difference Equations and Applications},
  11(4--5):379--398, 2005.

\bibitem{horn1972}
F~Horn and R~Jackson.
\newblock General mass action kinetics.
\newblock {\em Arch. Rational Mech. Anal.}, 47(2):81--116, Jan 1972.

\bibitem{kaufman2007}
M.~Kaufman, C.~Soul{\'e}, and R.~Thomas.
\newblock A new necessary condition on interaction graphs for
  multistationarity.
\newblock {\em J. Theor. Biol.}, 248:675--685, Jan 2007.

\bibitem{mincheva2007}
Maya Mincheva and Marc~R Roussel.
\newblock {Graph-theoretic methods for the analysis of chemical and biochemical
  networks. I. Multistability and oscillations in ordinary differential
  equation models.}
\newblock {\em J. Math. Biol.}, 55(1):61--86, Jul 2007.

\bibitem{savageau1969}
M.A. Savageau.
\newblock Biochemical systems analysis: {I.} {S}ome mathematical properties of
  the rate law for the component enzymatic reactions.
\newblock {\em Journal of theoretical biology}, 25(3):365--369, 1969.

\bibitem{savageau1969a}
M.A. Savageau.
\newblock Biochemical systems analysis: {II.} {T}he steady-state solutions for
  an n-pool system using a power-law approximation.
\newblock {\em Journal of theoretical biology}, 25(3):370--379, 1969.

\bibitem{savageau1970}
M.A. Savageau.
\newblock Biochemical systems analysis: {III.} {D}ynamic solutions using a
  power-law approximation.
\newblock {\em Journal of theoretical biology}, 26(2):215--226, 1970.

\bibitem{smith1988}
Hal~L. Smith.
\newblock Systems of ordinary differential equations which generate an order
  preserving flow. a survey of results.
\newblock {\em SIAM Review}, 30(1):87--113, 1988.

\bibitem{sontag2007}
E~Sontag.
\newblock Monotone and near-monotone biochemical networks.
\newblock In {\em Lecture Notes in Control and Information Sciences}, volume
  357, pages 79--122, Jan 2007.

\bibitem{soule2003}
Christophe Soule.
\newblock Graphic requirements for multistationarity.
\newblock {\em ComplexUs}, 1:123--133, 2003.

\bibitem{thomas2001}
R.~Thomas and M.~Kaufman.
\newblock Multistationarity, the basis of cell differentiation and memory. i.
  structural conditions of multistationarity and other nontrivial behavior.
\newblock {\em Chaos}, 11(1):170--179, Mar 2001.

\end{thebibliography}
\end{document}